\documentclass[11pt]{article}

\usepackage[hidelinks]{hyperref}
\usepackage{fullpage}
\usepackage{amsmath,amsthm,amsfonts,amssymb}
\usepackage{physics}
\usepackage{xcolor}
\usepackage{tikz}

\newtheorem{lemma}{Lemma}
\newtheorem{theorem}{Theorem}
\newtheorem{corollary}{Corollary}

\newcommand{\MAXCUT}{\textsc{Max-Cut}}
\newcommand{\QMAXCUT}{\textsc{Quantum Max-Cut}}
\newcommand{\PROD}{\textsc{Product-QMC}}

\newcommand{\QMA}{\textup{QMA}}
\DeclareMathOperator{\poly}{poly}
\DeclareMathOperator{\OPT}{OPT}
\DeclareMathOperator{\OPTprod}{OPT_{prod}}

\newcommand{\R}{\mathbb{R}}
\newcommand{\C}{\mathbb{C}}
\newcommand{\T}{\textsf{T}}

\renewcommand{\S}{\hat{S}}
\renewcommand{\vec}[1]{\mathbf{#1}} 
\newcommand{\vvec}[1]{\overrightarrow{\mathbf{#1}}}

\makeatletter
\newtheorem*{rep@theorem}{\rep@title}
\newcommand{\newreptheorem}[2]{%
\newenvironment{rep#1}[1]{%
 \def\rep@title{#2 \ref{##1}}%
 \begin{rep@theorem}}%
 {\end{rep@theorem}}}
\makeatother

\newreptheorem{theorem}{Theorem}

\title{Quantum Max-Cut is NP-hard to approximate}
\author{Stephen Piddock \\ \textit{Department of Computer Science, Royal Holloway, London, UK}}

\begin{document}
\maketitle

\begin{abstract}
    We unconditionally prove that it is NP-hard to compute a constant multiplicative approximation to the \QMAXCUT{} problem on an unweighted graph of constant bounded degree. The proof works in two stages: first a generic reduction to computing the optimal value of a quantum problem, from the optimal value over product states. Then an approximation preserving reduction from \MAXCUT{} to \PROD{}, the product state version of \QMAXCUT.
   
    More precisely, in the second part, we construct a PTAS reduction from \MAXCUT$_k$ (the rank-$k$ constrained version of \MAXCUT) to \MAXCUT$_{k+1}$, where \MAXCUT{} and \PROD{} coincide with \MAXCUT$_1$ and \MAXCUT$_3$ respectively. We thus prove that \MAXCUT$_k$ is APX-complete for all constant k.

\end{abstract}

%No macros abstract:
%We unconditionally prove that it is NP-hard to compute a constant multiplicative approximation to the QUANTUM MAX-CUT problem on an unweighted graph of constant bounded degree. The proof works in two stages: first we demonstrate a generic reduction to computing the optimal value of a quantum problem, from the optimal value over product states. Then we prove an approximation preserving reduction from MAX-CUT to PRODUCT-QMC the product state version of QUANTUM MAX-CUT.
%More precisely, in the second part, we construct a PTAS reduction from MAX-CUT_k (the rank-k constrained version of MAX-CUT) to MAX-CUT_{k+1}, where MAX-CUT and PRODUCT-QMC coincide with MAX-CUT_1 and MAX-CUT_3 respectively. We thus prove that Max-Cut_k is APX-complete for all constant k.

\section{Introduction}

\QMAXCUT{} is a quantum generalisation of the classical \MAXCUT{} problem . In \MAXCUT, the aim is to assign values $s_i \in\{\pm 1\}$ for each vertex in a graph $G$ in order to maximise the number of edges $ij$ with different signs $s_i \neq s_j$.  
More precisely, given a weighted graph $G=(V,E,w)$, with $|V|=n$ and $|E|=m$
\[\MAXCUT(G)=\max_{s \in \{\pm 1\}^n }\sum_{ij\in E} w_{ij}\frac{1}{2}(1-s_i s_j).\] 
\QMAXCUT{} is similar except that there is qubit at each vertex of the graph and the problem is to compute $\OPT(H)$, the maximum energy $\Tr(H\rho)$ over $n$-qubit quantum states $\rho$ with respect to the Hamiltonian $H$:

\[H=\sum_{ij \in E}w_{ij}\frac{1}{4}(I-X_iX_j -Y_iY_j - Z_i Z_j).\] 
where $X,Y,Z$ are the Pauli matrices and the subscript $i$ denotes the action of an operator on the $i$th qubit. 
Here the weight $w_{ij}$ is interpreted as the strength of the interaction $I-X_iX_j -Y_iY_j - Z_i Z_j$.

%Heisenberg interaction
Computing $\OPT(H)$ exactly is equivalent to calculating the maximum eigenvalue of $H$, or the ground state energy of $-H$, an antiferromagnetic Heisenberg Hamiltonian. 
The ground state energy is the physically relevant quantity, and the ``Local Hamiltonian problem'' is the task of computing the ground state energy to inverse polynomial accuracy. Here we focus on the maximisation formulation as it mirrors the classical \MAXCUT{} problem, and is better suited to approximation algorithms.

As the only two local qubit interaction with $SU(2)$ symmetry (i.e. it is invariant under conjugation by $U\otimes U$ for any $U \in SU(2)$), the Heisenberg interaction is in a sense one of the simplest and most natural two-qubit interactions. 
It is closely related to the Fermi-Hubbard model at half-filling, a simple but important model of quantum magnetism.

Due to its symmetry, the Heisenberg interaction plays a central part in the proof of the complexity classification of local Hamiltonian problems \cite{CubittMontanaro16}, and in the construction of universal qubit Hamiltonians \cite{Cubitt18}.
The decision variant of \QMAXCUT{} is \QMA-complete \cite{PM17}, where the aim is to decide if $\OPT(H)$ is above $a$ or below $b$ for thresholds $a,b\in \R$ with $a-b\ge 1/\poly(n)$.
A major drawback of this result is that the weights $w_{ij}$ vary throughout the graph, becoming as large as $\poly(n)$. This distances the model from the most physically relevant setting where the weights are all equal to 1.

Since it is hard to compute \QMAXCUT{} to additive error, attention has turned to approximation algorithms \cite{GP19,parekh2022optimal,lee2022optimizing,king2022improved,lee2024improved}. 
These algorithms do not find the true optimal solution, but find states $\rho$ which are guaranteed to achieve an energy $\Tr(H\rho)$ that comes within a constant multiplicative factor $\alpha$ of the optimal value $\QMAXCUT(G)$.

The approximation algorithm of \cite{GP19} is based on the famous Goemans-Williamson approximation algorithm \cite{goemans1995improved} for \MAXCUT, and its generalisation \cite{briet2010positive}. 
For \MAXCUT, the Goemans-Williamson algorithm is known to be optimal in the sense that, assuming the unique games conjecture, it is NP-hard to get a better approximation ratio for \MAXCUT{} \cite{khot2007optimal}.
Even without the unique games conjecture, \MAXCUT{} is known to be APX-complete \cite{papadimitriou1988optimization}, meaning it is as hard as any problem in APX, the class of classical optimisation problems with an efficiently achievable constant approximation ratio. And H{\aa}stad proved that it is NP-hard to approximate \MAXCUT{} for any approximation ratio better than $16/17$ \cite{haastad2001some}. 

For \QMAXCUT, the situation is very different. A number of extensions and improvements have been developed since the algorithm of \cite{GP19}, with the best achievable approximation ratio increasing over time \cite{parekh2022optimal,lee2022optimizing,king2022improved,lee2024improved}, and the simple SDP based method of \cite{GP19} proving not to be optimal.

On the hardness side, we currently have no way to extend the QMA-hardness result of \cite{PM17} to the regime of constant approximation ratio algorithms, since this would require proving a quantum version of the PCP theorem, which appears to be beyond the reach of current tools.
Although \QMAXCUT{} is introduced as a generalisation of \MAXCUT, we do not have any way to directly embed \MAXCUT{} instances into \QMAXCUT{} instances, and so we do not inherit any of the NP-hardness results from \MAXCUT. 
Recently, there has been an attempt \cite{hwang2022unique} at upper bounding the approximation ratio for \QMAXCUT{} assuming the unique games conjecture, analagous to \cite{khot2007optimal}.

Is \QMAXCUT{} hard when the interaction strengths are $O(1)$, or even uniform? Is it hard to approximate up to a constant approximation ratio? Can we prove a more direct reduction from \MAXCUT{} to \QMAXCUT ?
We remark that the first and third question here are the essence of open questions 3.2 and 3.3 from Gharibian's recent survey paper on quantum versions of NP \cite{Gharibian_2023}.

We resolve these questions in our main result:

\begin{theorem}
    \label{thm:hardtoapproximate}
    There is a constant $\alpha<1$ such that it is NP-hard to compute a value that approximates the optimal  \QMAXCUT{} value within approximation ratio $\alpha$.
    This holds even when the graph $G$ is unweighted ($w_{ij}=1$ for all $ij\in E$) and each vertex has bounded (constant) degree.
\end{theorem}

We believe that this result is the first hardness of approximation result for a family of quantum local Hamiltonian problems that does not require any conjectures or include a known hard-to-approximate classical CSP in that family.

Theorem~\ref{thm:hardtoapproximate} is proven by reducing from \MAXCUT{} to \QMAXCUT{}, although the reduction is not completely direct. As an intermediate problem, we consider the product state variant of \QMAXCUT.

First we show that approximating $\QMAXCUT(G)=\OPT(H)$ is at least as hard as approximating $\PROD(G)=\OPTprod(H)$, the maximal value over product states:
\[\PROD(G)=\OPTprod(H)=\max_{\rho =\rho_1 \otimes \rho_2\otimes \dots \otimes \rho_n} \Tr(H\rho)\]
We do this by taking a graph $G$, and replacing each vertex $i$ with a cloud of vertices $T_i$ where $|T_i|=T$, with edges between all vertices in $T_i$ and $T_j$ iff $ij \in E$. 
The new graph $G'$ has the same rough structure as $G$, but now each vertex has high degree, and since each qubit cannot be highly entangled with its neighbours, product states are now close to optimal. 

In fact the same idea applies generically to any $k$-local Hamiltonian and we prove the following theorem:
\begin{theorem}
    \label{thm:prodapprox}
    Let $H=\sum_{\vec{i}\in I} h_{\vec{i}}$ be a $k$-local qubit Hamiltonian with $h_{\vec{i}}\geq 0$ for all $\vec{i} \in I$. Then one can construct $H'$ using the interactions $\{h_{\vec{i}}\}$ such that
    \[T^k\OPTprod(H) \le \OPT(H')\le (T+2)^k\OPTprod(H)\]
\end{theorem}

We remark that Theorem~\ref{thm:prodapprox} has a very similar flavour to the results of Brandao and Harrow \cite{brandao}. Indeed, their Corollary 4 shows that for a Hamiltonian with an interaction graph of high degree, there is a product state which is nearly optimal.
With an additional argument relating the optimal product state for $H'$ and the optimal product state for $H$, one could prove a variant of Theorem~\ref{thm:prodapprox} using their result.  

However we instead provide a novel proof of Theorem~\ref{thm:prodapprox} based on the representation theory of $SU(2)$ and an inequality due to Lieb \cite{lieb1973classical}.
This has the advantage of being simple, easily extends the results to $k>2$ and achieves better constant factors.
The disadvantage of this method is that it apparently cannot be generalised to systems of qu\emph{d}its\footnote{Qudits are quantum subsystems with a complex state space of dimension $d$. For example, qubits are qudits of dimension $d=2$.}. 
%If one desired a version of Theorem~\ref{thm:prodapprox} for qudits, then the methods of \cite{brandao} may be required.

The recent work of \cite{kallaugher2024complexity} classifies the complexity of the product state problem for $\mathcal{S}$-Hamiltonians, the family of Hamiltonians consisting of interactions from a fixed set $\mathcal{S}$ (with positive and negative weights allowed).  Therefore as a corollary of Theorem~\ref{thm:prodapprox}, we have:
\begin{corollary}
    Let $\mathcal{S}$ be a set of two qubit interactions. Then the Local Hamiltonian problem for $\mathcal{S}$-Hamiltonians (i.e. determining if the ground state energy is below $a$ or above $b$ for $b-a\ge 1/\poly(n)$) is:
    \begin{enumerate}
        \item in P, if all elements of $\mathcal{S}$ are 1-local.
        \item NP-hard otherwise. Furthermore, this holds even when the interaction strengths are $O(1)$.
    \end{enumerate}
    
\end{corollary}
Note that this corollary is weaker than the classification of \cite{CubittMontanaro16}, except that here the interaction strengths are constant. 

Armed with Theorem~\ref{thm:prodapprox}, our focus in the second part of the paper is to prove hardness of approximation for finding the optimal product state $\rho_1\otimes \rho_2 \otimes \dots \otimes \rho_n$ for $\PROD(G)$. 
By decomposing a single qubit density matrix in the Pauli basis as $\rho=\frac{1}{2}(I+x_1X+x_2Y+x_3Z)$, we can identify $\rho$ with the 3-dimensional unit vector of Pauli coefficients $\vec{x}$, and the product state $\rho_1\otimes \rho_2 \otimes \dots \otimes \rho_n$ with the list of unit vectors $\vvec{x}=(\vec{x_1},\vec{x_2},\dots \vec{x_n})$.

The rank-$k$ constrained \MAXCUT{} problem, which we denote $\MAXCUT_k$ is a generalisation of $\MAXCUT$, where each vertex is assigned a unit vector in $\R^k$, i.e. an element of the sphere $S_{k-1}$.

\[ \MAXCUT_k(G)= \max_{\vvec{x}=(\vec{x_1},\vec{x_2},\dots, \vec{x_n})\in (S_{k-1})^n }  \sum_{ij \in E} \frac{1}{2}(1-\vec{x_i}\cdot \vec{x_j})\] 
Therefore, as has been observed previously (for example in \cite{GP19}), $\MAXCUT_3(G)=2\times \PROD(G)$; and note that $\MAXCUT=\MAXCUT_1$.
Bri{\"e}t, de Oliveira Filho and Vallentin showed how the Goeman-Williams approximation algorithm for \MAXCUT{} can be generalised to $\MAXCUT_k$ for every $k$ \cite{briet2010positive}. 
Strong evidence was provided in \cite{hwang2022unique} that  the approximation algorithms of \cite{briet2010positive} are optimal: they proved that, assuming the unique games conjecture, and a newly conjectured inequality referred to as the vector-valued Borel inequality, it is NP-hard to approximate $\MAXCUT_k$ to any better approximation ratio.

However, without these assumptions, the hardness of $\MAXCUT_k$ is less well understood. While $\MAXCUT_1$ was on Karp's original list of NP-complete problems \cite{Karp1972}, and Lov{\'a}sz and others have conjectured that $\MAXCUT_k$ is hard for other values of $k$ (see e.g. p236 of \cite{lovasz2019graphs}), it is only very recently that $\MAXCUT_3$ was proven to be NP-complete for additive inverse polynomial accuracy \cite{kallaugher2024complexity}, and no other hardness results are known for $k\neq 1,3$. 

Furthermore no hardness of approximation results are known without additional assumptions, except for $k=1$, where it is known that \MAXCUT{} is APX-complete \cite{papadimitriou1988optimization}. We extend these results to all $k>1$ in Theorem~\ref{thm:PTAS}.

\begin{theorem}
    \label{thm:PTAS}
    There is a PTAS reduction from $\MAXCUT_{k}$ to $\MAXCUT_{k+1}$. 
    Therefore $\MAXCUT_k$ is APX-complete for all (constant) $k\ge 1$.
\end{theorem}

We remark that $\MAXCUT_2 = 2\times\PROD_{XY}(G)$, the problem of finding the optimal product state for the Hamiltonian $H_{XY}$:
\[H_{XY}=\sum_{ij \in E} w_{ij} \frac{1}{4}(I-X_iX_j -Y_iY_j)\]
and thus we also get a hardness of approximation result for $\PROD_{XY}(G)$.

To reduce from $\MAXCUT_{k}$ to $\MAXCUT_{k+1}$, it is necessary to find a way to effectively reduce the available dimensions at each vertex in a $\MAXCUT_{k+1}$ instance. One natural way to do this is to look for a gadget which forces two vertices to be orthogonal (or approximately orthogonal). Then by selecting a reference vertex $0$, and applying this gadget to all pairs $i,0$, one can ensure that $\vec{x_i}$ is orthogonal to $\vec{x_0}$ for all $i$, and thus the $\vec{x_i}$ are effectively in a lower dimensional space. Wright \cite{Wrightpersonal} had suggested a gadget based on a long odd length cycle for this task.

Our insight is to recognise that it is not necessary to make $\vec{x_i}$ and $\vec{x_0}$ orthogonal, but it is enough to make sure that $\vec{x_i}$ and $\vec{x_0}$ are at a fixed angle from each other for all $i$. For this task, almost any gadget will do, as long as it is not a bipartite graph, and we focus on the simplest example of a triangle. 
In a triangle graph, the optimal assignment for $\MAXCUT_k$ has an angle of $2\pi/3$ between each of the vertices on the triangle. 

We sketch a simplified version of our construction based on this idea. Take an instance $G$ of $\MAXCUT_k$ and add a vertex $c_i$ for each $i$ and an extra vertex $0$. Include the original edges of $G$, and add heavily weighted edges between the vertices of the triangle $0,i,c_i$ for all $i$ to create a new graph $G'$.
The heavily weighted triangle gadgets mean that any assignment $\vvec{x}$ close to optimal must have $\vec{x_i}=\cos(\theta_i)\vec{x_0}+\sin(\theta_i)\vec{\tilde{x}_i}$ for some $\vec{\tilde{x}_i}$ orthogonal to $\vec{x_0}$, and $\theta_i \approx 2\pi/3$.
Since the $\vec{\tilde{x}_i}$ are orthogonal to $\vec{x_0}$, we can interpret them living in $\R^{k}$, rather than $\R^{k+1}$.

Therefore the optimal value $\MAXCUT_{k+1}(G')$ will be approximately equal to a simple function of $\MAXCUT_k(G)$. Furthermore by projecting the vector $\vec{x_i}$ onto the space orthogonal to $\vec{x_0}$ and renormalising for all $i$, we have a map $g:(S_k)^{n+1}\rightarrow (S_{k-1})^n$ which maps good solutions for $\MAXCUT_{k+1}(G')$ to good solutions for $\MAXCUT_k(G)$.
This is essential for showing that the construction is a PTAS reduction.

While this sketch can be formalised into a PTAS reduction and hence prove $\MAXCUT_k$ is APX-hard, there are a couple of drawbacks of this simple approach. Firstly, the heavily weighted terms are undesirable: although they only need to be large constant values (not scaling with $n$, the number of vertices), we want a construction where all weights are equal to 1. Second, the reference vertex $0$ has very high degree as it is connected to all other vertices.

We overcome these difficulties by replacing the single reference vertex $0$ with a large bipartite expander graph of additional vertices. The optimal assignment for a bipartite graph is to have all vertices on one side of the bipartition set to the same vector $\vec{z}$ and all vertices on the other side of the partition set to $-\vec{z}$. Rather than have a heavily weighted triangle between $i$ and $0$, we instead have multiple unweighted triangles between $i$ and different vertices from one side of the bipartite graph. This then reproduces the behaviour of the simple reduction but without any weighted edges or high degree vertices.

We briefly contrast our proof method to that of \cite{kallaugher2024complexity}, where in their Theorem 1.5, they prove that it is NP-hard to determine the optimal value  of $\MAXCUT_3$ up to additive inverse polynomial accuracy. Rather than reduce from $\MAXCUT$, their method instead reduces from $3$-\textsc{COLORING}. The construction involves taking an instance $G$ of $3$-\textsc{COLORING}, and replacing each edge $ij$ with a clique of four vertices, then taking that new graph and replacing each edge with a clique of three vertices. Finally an additional vertex is added and connected to some of the vertices. This final vertex plays a similar role to the vertex $0$ described in the sketch of our construction above, in that it provides a reference state to fix a basis for the other vertices. 

Our proof method is arguably simpler, provides a direct connection to \MAXCUT, and obtains the stronger result of APX-completeness for $\MAXCUT_3$.

\section{Preliminaries}

\subsection{SU(2) representation theory}
%The proof of Theorem~\ref{thm:prodapprox} relies heavily on an inequality due to Lieb \cite{lieb1973classical}, and some basic representation theory of $SU(2)$, which we introduce here.

The Lie group $SU(2)$ is the group of unitary $2\times 2$ matrices with determinant $1$. The associated Lie algebra $\mathfrak{su}(2)$ is the space of traceless anti-hermitian $2\times 2$ matrices. 
Equivalently, any element $A\in \mathfrak{su}(2)$ can be decomposed as $A=a_1\sigma^{(1)}+a_2\sigma^{(2)}+a_3\sigma^{(3)}$ for real numbers $a_1,a_2,a_3$ and where the $\sigma^{(j)}$ are the Pauli matrices:
\[
\sigma^{(1)}=X=\begin{pmatrix}
    0 &1 \\
    1 &0 \\
\end{pmatrix},
\quad 
\sigma^{(2)}=Y=\begin{pmatrix}
    0 &-i \\
    i &0 \\
\end{pmatrix},
\quad 
\sigma^{(3)}=Z=\begin{pmatrix}
    1 &0 \\
    0 &-1 \\
\end{pmatrix}.\]
For convenience, we will also write $\sigma^{(0)}$ for the $2\times 2$ identity matrix $\sigma^{(0)}=I=\left(\begin{smallmatrix}
    1 &0 \\
    0 &1 \\
\end{smallmatrix}\right)$, so that $\{\sigma^{(j)}\}_{j=0}^4$ is a basis for all hermitian $2\times 2$ matrices.

Note that $\mathfrak{su}(2)$ is closed under the action of the matrix commutator (sometimes referred to as the Lie bracket) $[A,B]=AB-BA$.

A \emph{representation} $R$ of dimension $d$ is a map $R:\mathfrak{su}(2) \rightarrow M(\C^d)$, where $M(\C^d)$ is the space of complex square matrices of dimension $d$, such that 
\[[R(A),R(B)]=[A,B] \qquad\qquad \text{ for all }  A,B \in \mathfrak{su}(2).\]
The representation $R$ is \emph{irreducible} if there is no non-trivial subspace  $W\subset\C^d$ for which $R(A)$ maps $W$ to itself for all $A\in\mathfrak{su}(2)$.

For each positive integer $d$, there exists a unique (up to isomorphism) irreducible representation of $\mathfrak{su}(2)$ of dimension $d$.  We denote this $d$-dimensional irreducible representation $R_J$ where $J=\frac{d-1}{2}$, and will refer to it as the  \emph{spin-$J$ representation}, matching the terminology common in physics.

The Lie algebra $\mathfrak{su}(2)$ is \emph{semi-simple} which means that any representation $R$ is isomorphic to a direct sum of irreducible representations. 
Concretely, for a representation $R$ there is a unitary basis change $U$ on $\C^d$ such that:
\begin{equation}R(A)=U\left(\bigoplus_{J}R_J(A)\right)U^{\dagger} \qquad \qquad \text{ for all } A\in \mathfrak{su}(2).
    \label{eqn:semisimple}
\end{equation}

Given two representations $R_1$ and $R_2$ one can define the \emph{tensor product representation} $R_1\otimes R_2$ on $\C^{d_1}\otimes \C^{d_2}$:
\[R_1\otimes R_2(A)=R_1(A)\otimes I_{d_2} +I_{d_1}\otimes R_2(A)\]

We are particularly interested in the representation $\hat{R}$ produced by taking the tensor product of $T$ copies of the spin-$\frac{1}{2}$ representation. 
This acts on the space $(\C^2)^{\otimes T}$, which we can identify with the space of $T$ qubits:
\[\hat{R}(A)=\sum_{t=1}^T A_t\]
where $A_t$ denotes the operator $A$ acting on qubit $t$, and identity elsewhere.

The decomposition of $\hat{R}$ into irreducible representations as in equation \eqref{eqn:semisimple} is very well understood. In fact, it is best described via Schur-Weyl duality, which connects to the representation theory of the symmetric group.
However all we need is that there exists a unitary $U$ such that 
\begin{equation}
    \hat{R}(A)=U\left(\bigoplus_{J\in \mathcal{T}}R_J(A)\otimes I_{g_J}\right)U^{\dagger} \qquad \qquad \text{ for all } A\in \mathfrak{su}(2).
    \label{eq:hatRdecomp}
\end{equation}
where $I_{g_{J}}$ is the identity matrix on a space of dimension $g_J$ and $\mathcal{T}=\{T/2,T/2-1,T/2-2,\dots\}$ (continuing down to $0$ if $T$ is even, or $1/2$ if $T$ is odd).

\subsection{Local Hamiltonians and Lieb's inequality}
\label{sec:LHamLieb}
Let $H$ be a $k$-local Hamiltonian on $n$ qubits. Let $I$ be a subset of $[n]^k$ so that $H$ can be written as $H=\sum_{\vec{i} \in I}h_{\vec{i}}$ for interactions $h_{\vec{i}}$ which act non trivially only on the qubits in $\vec{i}=(i_1,i_2,\dots i_k)$.
Each $h_{\vec{i}}$ can be written in the Pauli basis as 
\begin{equation}
    h_\vec{i}=\sum_{\vec{a}\in \{0,1,2,3\}^k} M(\vec{i})_{a_1,a_2,\dots a_k}\sigma_{i_1}^{(a_1)}\sigma_{i_2}^{(a_2)} \dots \sigma_{i_k}^{(a_k)}
    \label{eq:hpaulibasis}
\end{equation}
for some order $k$ tensor $M(\vec{i})$.

To introduce Lieb's inequality, we define $H(\vec{J})$ and $\tilde{H}(\vec{J})$ two variants of $H$ , both of which depend on a vector $\vec{J}=(J_1,J_2, \dots , J_n) \in (\mathbb{Z}_{>0}/2)^n$.

First to define $H(\vec{J})$, we imagine that qubit $i$ is replaced by a spin-$J_i$ particle (i.e. a qudit of dimension $d_i=2J_i+1$).
 Then $H(\vec{J})$ looks the same as $H$, except that each $\sigma^{(a)}_i$ is replaced with $R_{J_i}(\sigma^{(a)})_i$.

 \begin{equation}
    \label{eq:HJ}
    H(\vec{J})=\sum_{\vec{i}\in I}\sum_{\mathbf{a}\in \{0,1,2,3\}^k} M(\vec{i})_{a_1,a_2,\dots a_k} R_{J_{i_1}}(\sigma^{(a)})_{i_1} R_{J_{i_2}}(\sigma^{(a)})_{i_2} \dots R_{J_{i_k}}(\sigma^{(a)})_{i_k}
 \end{equation}

The other variant $\tilde{H}(\vec{J})$ still acts on $n$ qubits and has the same structure as the original Hamiltonian $H$, but each $\sigma_{i}^{(a)}$ is replaced by $2J_{i}\sigma_{i}^{(a)}$. Equivalently $\tilde{H}(\vec{J})$ can be defined as:
 \begin{equation}
    \label{eq:tildeHJ}
    \tilde{H}(\vec{J})=\sum_{\vec{i}\in I}\left(\prod_{l=1}^k{2 J_{i_l}}\right)h_{\vec{i}}.
 \end{equation}

 An inequality of Lieb (equation 5.6 of \cite{lieb1973classical}) tells us that 

 \begin{equation}
     \label{eq:Liebinequality}
     \OPT(H(\vec{J}))\le \OPTprod(\tilde{H}(\vec{J}+\vec{1}))  
 \end{equation}
 where $\vec{1}$ represents the all ones vector.\footnote{The notation in \cite{lieb1973classical} does not match our notation here. The Hamiltonian $H(\vec{J})$ is the main object of study in \cite{lieb1973classical}, and their equation 5.6 is expressed in terms of the ground state energy  (lowest eigenvalue), which they denote $E^Q$, so the inequality appears in the reverse direction. Their notation $E^C$ is for the lowest energy of the classical Hamiltonian which we may think of as the product state energy of $\tilde{H}(\vec{J})$}

\subsection{The APX class and PTAS reductions}

An NP optimisation (NPO) problem is a triple $(I,S,m)$ where:
\begin{enumerate}
    \item $I$ is a set of instances.
    \item Given an instance $x \in I$, the set $S(x)$ is the set of valid solutions to $x$.
    \item Given an instance $x \in I$ and a valid solution $y \in S(x)$, the value achieved by $y$ is $m(x,y)\in \R$. 
\end{enumerate}
The sets $I$ and $S(x)$ are subsets of $\{0,1\}^{*}$ which can be recognised in polynomial time, and the function $m$ is computable in polynomial time.

An NP optimisation problem can be a either a maximisation problem or a minimisation problem. Given an instance $x$ of a maximisation NPO problem, the aim is to find a valid solution $y$ which maximises the value $m(x,y)$.
An algorithm that, for any instance $x$, finds a $y$ such that $m(x,y)$ is at least $\alpha$ times the optimal value, is said to have approximation ratio $\alpha$.

An NP optimisation problem is in the class APX if there exists a constant $\alpha$ and a polynomial time algorithm $T$ that achieves an approximation ratio $\alpha$. The class PTAS is a subset of APX, and consists of all NPO problems for which there exists a polynomial time approximation scheme. That is, a problem is in PTAS if there exists an algorithm $T$ such that for any instance $x$ and approximation ratio $\alpha$, $T(x,\alpha)$ returns a valid solution $y$ that achieves the approximation ratio $\alpha$, and $T$ runs in time $q_r(|x|)$ for some polynomial $q_r$.

When proving that a problem is hard for the class APX, it is necessary to use PTAS reductions, which are a type of approximation-preserving reduction which preserves containment in the class PTAS. One reason for this is that we want APX-hardness to imply that a problem does not have a polynomial time approximation scheme (assuming that APX$\neq$PTAS).

Formally a PTAS reduction from problem A to problem B is a tuple of functions $(f,g,c)$ such that, for any instance $x$ of A and $\alpha<1$:
\begin{enumerate}
    \item $f(x,\alpha)$ is an instance of $B$ computable in polynomial time in $|x|$.
    \item for any valid solution $y$ to $f(x,\alpha)$, then $g(x,y,\alpha)$ is a valid solution of $x$, computable in polynomial time in $|x|$ and $|y|$.
    \item for any valid solution $y$ to $f(x,\alpha)$, if $y$ achieves an approximation ratio at least $c(\alpha)$ for the instance $f(x,\alpha)$, then $g(x,y,\alpha)$ achieves an approximation ratio at least $\alpha$ for the instance $x$.
\end{enumerate}

Therefore a PTAS reduction provides a map between instances of A and B with related optimal values, but also map back such that `good' solutions of the B instance are mapped back to `good' solutions of the A instance. 
Note that a PTAS reduction from A to B provides an algorithm for A with approximation ratio $\alpha$, if there is an algorithm for B which achieves an approximation ratio $c(\alpha)$.
If $c$ is not the identity function, then it may be that the approximation ratio $c(\alpha)$ required for B is a lot higher than the desired approximation ratio $\alpha$ for A.

Other (simpler) kinds of approximation-preserving reductions are often used when proving APX-completeness results, such as L-reductions or AP-reductions, see~\cite{approxreductions} for more. 
However the reduction we use in Lemma~\ref{lem:main} constructs a different instance $G'$ depending on a parameter $\eta$ which is chosen to achieve particular approximation ratios $\alpha$. That is, the function $f(x,\alpha)$ in the PTAS reduction really does depend on $\alpha$, which is not allowed in the other types of approximation-preserving reductions.

\section{Reduction from product state problem}

In this section we prove Theorem~\ref{thm:prodapprox}, which relates the optimal product state value of a local Hamiltonian $H$ to the optimal value of another Hamiltonian $H'$, constructed from the same interactions as $H$. First we describe how to construct $H'$ from $H$.

Using the notation introduced in Section~\ref{sec:LHamLieb}, we write $H=\sum_{\vec{i} \in I}h_{\vec{i}}$ for interactions $h_{\vec{i}}$ which act non trivially only on the qubits in $\vec{i}=(i_1,i_2,\dots i_k)$, and can be expressed in the Pauli basis as in equation \eqref{eq:hpaulibasis}, so that 
\begin{equation}
    H
    =\sum_{\vec{i}\in I}\sum_{\mathbf{a}\in \{0,1,2,3\}^k} M(\vec{i})_{a_1,a_2,\dots a_k}\sigma_{i_1}^{(a_1)}\sigma_{i_2}^{(a_2)} \dots \sigma_{i_n}^{(a_k)}
    %&=\sum_{\vec{i}}\sum_{\mathbf{a}\in \{0,1,2,3\}^k} M(\vec{i})_{a_1,a_2,\dots a_k}\left(\sum_{t_1\in [T]}\sigma_{(i_1,t_1)}^{(a_1)}\right)\sigma_{(i_2,t_2)}^{(a_2)} \dots \sigma_{(i_k,t_k)}^{(a_k)}\\
\end{equation}

We construct $H'$ by replacing each qubit of $H$ with a cluster of $T$ qubits for some parameter $T$ which is used to control the quality of approximation. We label the $nT$ qubits by $(i,t)$ for some $i\in [n]$ and $t \in [T]$. Each term $h_{\vec{i}}$ in $H$ is replaced with the same interaction applied to qubits $(i_1,t_1),(i_2,t_2),\dots (i_k,t_k)$ for all combinations of $t_1,t_2,\dots t_k$. The overall Hamiltonian $H'$ is therefore:
\begin{equation}
    H'
    =\sum_{\vec{t}\in [T]^k}\sum_{\vec{i}\in I}\sum_{\mathbf{a}\in \{0,1,2,3\}^k} M(\vec{i})_{a_1,a_2,\dots a_k}\sigma_{(i_1,t_1)}^{(a_1)}\sigma_{(i_2,t_2)}^{(a_2)} \dots \sigma_{(i_k,t_k)}^{(a_k)}.
\end{equation}

\begin{figure}
    \centering
    \includegraphics[scale=0.7]{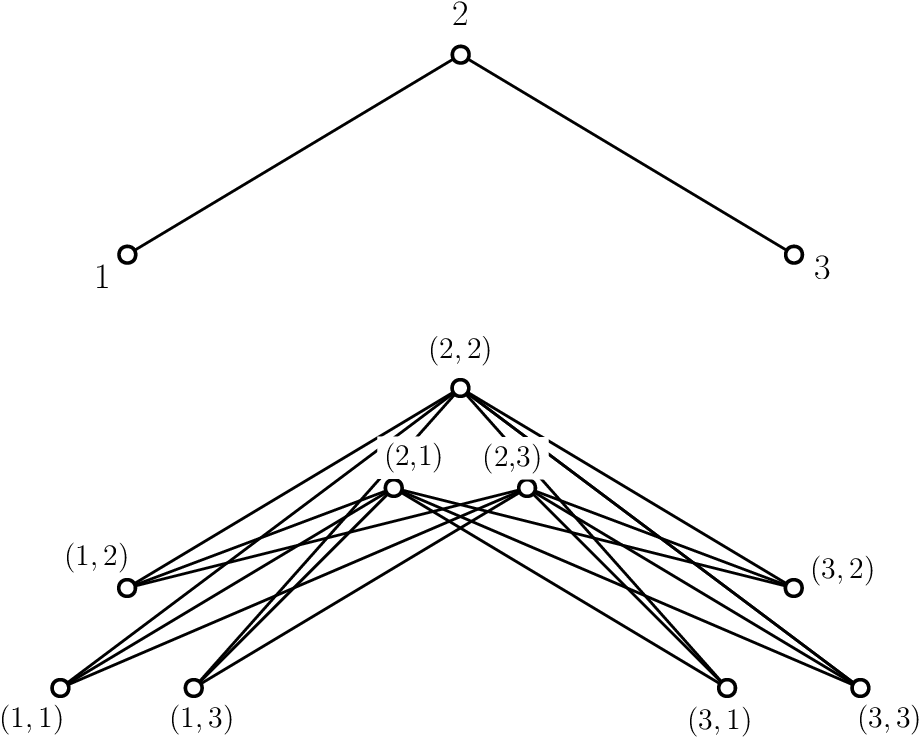}
    \caption{Example interaction graph of $H'$. At the top of the figure is a simple interaction graph of a 2-local Hamiltonian $H$ on $n=3$ qubits labelled $1,2,3$. Beneath is the interaction graph of the new Hamiltonian $H'$. Each qubit of $H$ has been replaced by a collection of $T=3$ qubits in $H'$.}
\end{figure}

%Having precisely described the construction, we can now state and prove Theorem~\ref{thm:prodapprox}.
\begin{reptheorem}{thm:prodapprox}
    Let $H=\sum_{\vec{i}\in I} h_{\vec{i}}$ be a $k$-local qubit Hamiltonian with $h_{\vec{i}}\geq 0$ for all $\vec{i} \in I$. Let $H'$ be constructed from $H$ as above. Then
    \[T^k\OPTprod(H) \le \OPT(H')\le (T+2)^k\OPTprod(H)\]
\end{reptheorem}

\begin{proof}
    First we prove the lower bound: let $\rho=\rho_1\otimes \rho_2 \otimes \dots \otimes \rho_n$ be the product state that achieves $\Tr(H\rho)=\OPTprod(H)$.
    Extend this to a product state $\rho'$ on $nT$ qubits where qubit $(i,t)$ is in state $\rho_i$ for all $t$. 
    Then 
    \[\OPT(H')\ge \Tr(H'\rho')=\sum_{\vec{t}\in [T]^k}\Tr(H\rho)=T^k \OPTprod(H).\]
    It remains to prove the upper bound.
    Let $\S^{(a)}_i=\sum_{t\in T} \sigma^{(a)}_{(i,t)}$ and observe that 
    \begin{align*}
        H'&=\sum_{\vec{i}\in I}\sum_{\mathbf{a}\in \{0,1,2,3\}^k} M(\vec{i})_{a_1,a_2,\dots a_k} \sum_{\vec{t}\in [T]^k}\sigma_{(i_1,t_1)}^{(a_1)}\sigma_{(i_2,t_2)}^{(a_2)} \dots \sigma_{(i_k,t_k)}^{(a_k)}\\
        &=\sum_{\vec{i}}\sum_{\mathbf{a}\in \{0,1,2,3\}^k} M(\vec{i})_{a_1,a_2,\dots a_k}\left(\sum_{t_1\in [T]}\sigma_{(i_1,t_1)}^{(a_1)}\right)\left(\sum_{t_2\in [T]}\sigma_{(i_2,t_2)}^{(a_1)}\right) \dots \left(\sum_{t_k\in [T]}\sigma_{(i_k,t_k)}^{(a_1)}\right)\\
        &= \sum_{\vec{i}\in I}\sum_{\mathbf{a}\in \{0,1,2,3\}^k} M(\vec{i})_{a_1,a_2,\dots a_k}  \S_{i_1}^{(a_1)}\S_{i_2}^{(a_2)} \dots \S_{i_k}^{(a_k)} .
    \end{align*}

    The operator $\S^{(a)}_i$ is the tensor product representation $\hat{R}$ (of $T$ copies of the spin $\frac{1}{2}$ representation) applied to $\sigma^{(a)}$, and it acts on the set of $T$ qubits $\{(i,t)\}_{t\in [T]}$ which replaced the original qubit $i$ in $H$.
    Therefore, as in equation~\eqref{eq:hatRdecomp} the tensor product representation $\hat{R}$ decomposes into a direct sum of irreducible representations:
    \[\S^{(a)}=U\left(\bigoplus_{J \in \mathcal{T}} R_{J}(\sigma^{(a)})\otimes I_{g_{J}}\right)U^{\dagger}\]
    for some unitary $U$, and where the sum is over all $J\in \mathcal{T}=\{T/2,T/2-1,T/2-2,\dots\}$ (continuing down to $0$ if $T$ is even, or $1/2$ if $T$ is odd).

    Therefore 
    \[H'=U^{\otimes n}\left(\bigoplus_{\vec{J}\in \mathcal{T}^n} H(\vec{J})\otimes I_{g_\vec{J}} \right) (U^{\dagger})^{\otimes n}\]
    where for each $\vec{J}=(J_1,J_2,\dots,J_n)\in \mathcal{T}^n$, $I_{g_{\vec{J}}}$ is the identity matrix on a space of dimension $g_{\vec{J}}=g_{J_1}\times g_{J_2}\dots \times g_{J_n}$, and $H(\vec{J})$ is the variant of $H$ where each qubit $i$ is replaced by a spin $J_i$ particle, as defined in equation \eqref{eq:HJ}.

     So the optimal energy for $H'$ is the maximum out of the optimal energies $\OPT(H(\vec{J}))$ out of all possible values of $\vec{J}\in \mathcal{T}^n$
    \begin{equation}
        \label{eq:H'maxHJ}
        \OPT(H')=\max_{\vec{J}\in \mathcal{T}^n}\OPT(H(\vec{J})).
    \end{equation}

    Applying Lieb's inequality (equation \eqref{eq:Liebinequality} above) to this gives:
    \begin{equation}
        \label{eq:H'maxHJ+1}
        \OPT(H') \le \max_{\vec{J}\in \mathcal{T}^n}\OPTprod(\tilde{H}(\vec{J}+\vec{1}))
    \end{equation}
    where $\vec{1}$ is the all ones vector and $\tilde{H}(\vec{J})$ is as defined in equation~\eqref{eq:tildeHJ}.

    Since $h_{\vec{i}}\geq 0 $ for all $\vec{i}$, increasing the value of $J_i$ for some $i$ can only increase the value of $\OPTprod(\tilde{H}(\vec{J}+\vec{1}))$. 
    %To see this, let $\tilde{\rho}$ be the product state that achieves $\OPTprod(\tilde{H}(\vec{J}))$
    Therefore the maximum in equation \eqref{eq:H'maxHJ+1} is obtained when each $J_i$ is maximised, i.e. at $\vec{J^*}=(T/2,T/2,\dots,T/2)\in \mathcal{T}^n$. For this $\vec{J^*}$, we have $\tilde{H}(\vec{J^*}+\vec{1})=(T+2)^k H$ and so 
    \[\OPT(H') \le \OPTprod(\tilde{H}(\vec{J^*}+\vec{1}))=(T+2)^k\OPTprod(H).\]

\end{proof}

\section{Rank constrained Max-Cut}

In this section we prove Theorem~\ref{thm:PTAS}, that $\MAXCUT_k$ is APX-complete, by demonstrating a PTAS reduction from $\MAXCUT_k$ to $\MAXCUT_{k+1}$.

We first introduce some notation for the cost function which is maximised in $\MAXCUT_{k}$. 
Let $F_k^G :(S_{k-1})^n \rightarrow \R$ be defined by:
\[F_k^G(\vvec{x})=\sum_{ij \in E} w_{ij}\frac{1}{2}(1-\vec{x_i}\cdot \vec{x_j})\] 
so that $\MAXCUT_k(G)=\max_{\vvec{x}}F_k^{G}(\vvec{x})$. 

In analogy to \QMAXCUT, we sometimes refer to $F_k^{G}(\vvec{x})$ as the energy of $\vvec{x}$. For the remainder we now take $w_{ij}=1$.

The main technical lemma which is used to prove Theorem~\ref{thm:PTAS} is Lemma~\ref{lem:main}:
\begin{lemma}
    \label{lem:main}
Given an instance $G$ of $\MAXCUT_k$ and an error parameter $\eta$, we construct an instance $G'$ of $\MAXCUT_{k+1}$ which is only polynomilally larger than $G$ and prove that there exist constants $c,C$ such that:
\begin{enumerate}
    \item $\MAXCUT_{k+1}(G') \ge C\eta m +\frac{3}{4} \MAXCUT_k(G)$
    \item there is a map $g:(S_{k})^{n+6\eta m}\rightarrow (S_{k-1})^n$ such that for any $\vvec{y} \in (S_{k})^{n+6\eta m}$:
    \[F_{k+1}^{G'}(\vvec{y}) \le C\eta m +\frac{cm}{\eta}+\frac{3}{4}F_k^{G}(g(\vvec{y})).\]
\end{enumerate}
\end{lemma}

We first prove how Theorem~\ref{thm:PTAS} follows from Lemma~\ref{lem:main}, delaying the proof of Lemma~\ref{lem:main} until Section~\ref{sec:lemproof}.

\begin{reptheorem}{thm:PTAS}
    There is a PTAS reduction from $\MAXCUT_{k}$ to $\MAXCUT_{k+1}$. 
    Therefore $\MAXCUT_k$ is APX-complete for all (constant) $k\ge 1$.
\end{reptheorem}

\begin{proof}
    To prove that the construction of Lemma~\ref{lem:main} is a PTAS reduction, we need to show that for any desired $\beta<1$, we can choose $\eta$ and $\alpha$ such that:
    \[F_{k+1}^{G'}(\vvec{y}) \ge \alpha \MAXCUT_{k+1}(G') 
    \quad \Rightarrow \quad F_k^G(g(\vvec{y}))\ge \beta \MAXCUT_k(G)\]
    So suppose that $F_{k+1}^{G'}(\vvec{y})\ge \alpha\MAXCUT_{k+1}(G')$ for some $\alpha$ to be chosen later. Using both parts of Lemma~\ref{lem:main}, we have 
    \begin{align*}
        C\eta m +\frac{cm}{\eta}+\frac{3}{4}F_k^{G}(g(\vvec{y}))
         & \ge F_{k+1}^{G'}(\vvec{y}) \\
         & \ge \alpha\MAXCUT_{k+1}(G') \\
         &\ge \alpha(C\eta m +\frac{3}{4} \MAXCUT_k(G))
    \end{align*}
    Rearranging this gives
    \begin{align*}
        F_k^{G}(g(\vvec{y})) &\ge \alpha\MAXCUT_{k}(G) +[(\alpha-1)C\eta -c/\eta]4m/3\\
        &\ge \left[\alpha+\frac{8}{3}(\alpha-1)C\eta-\frac{8c}{3\eta}\right]\MAXCUT_k(G)
    \end{align*}
    where we've used the fact that $\MAXCUT_k(G)\ge m/2$ and $[(\alpha-1)C\eta -c/\eta]<0$.
    Therefore it suffices to choose 
    \[\eta\ge \frac{16c}{3(1-\beta)}\quad \text{ and }\quad \alpha =  \frac{\frac{1+\beta}{2}+\frac{8C}{3}}{1+\frac{8C}{3}}\]
    so that $\alpha+\frac{8}{3}(\alpha-1)C\eta =\frac{1+\beta}{2}$ and $\frac{8c}{3\eta}\le (1-\beta)/2$.
\end{proof}

%\subsection{A single edge}
%Consider a graph $G$ consisting of a single edge between two vertices $1$ and $2$. 
%Then $F_k^G(\vec{x_1},\vec{x_2})=\frac{1}{2}(1-\vec{x_1}\cdot \vec{x_2})=\frac{1}{2} (1-\cos \theta)$ where $\theta$ is the angle between $\vec{x_1}$ and $\vec{x_2}$.

%Therefore if $\vec{x_1}$ and $\vec{x_2}$ change by an angle $\phi_1$ and $\phi_2$ respectively, then $\theta$ can change by at most $|\phi_1|+|\phi_2|$.

Before we prove Lemma~\ref{lem:main}, we describe the construction of $G'$ from $G$ and study the separate parts of $G'$, in order to build up enough understanding to complete the proof of Lemma~\ref{lem:main} in Section~\ref{sec:lemproof}.

\subsection{Construction}
\label{sec:construction}

Given a graph $G=(V,E)$ with $n=|V|$ vertices and $m=|E|$ edges, we extend $G$ to a graph $G'$ on vertex set $V\cup V'$ where $|V'|=6\eta m$ for some parameter $\eta$ which controls the quality of approximation.
For each vertex $i\in V$ of degree $d_i$, we add $3\eta d_i$  vertices $a_{i,\alpha},b_{i,\alpha},c_{i,\alpha}$ for $\alpha\in[\eta d_i]$.
There are three types of edges in $E'$:
\begin{enumerate}
    \item Edges in $E$ from the original graph $G$ between the vertices in $V$.
    \item There is a bipartite expander graph between the `$a$' vertices in $A=\{a_{i,\alpha}\}$ and the `$b$' vertices $B=\{b_{i,\alpha}\}$. This is chosen to be a $d$-regular expander graph for constant $d$, with constant spectral gap in the adjacency matrix.
    \item For each $i,\alpha$, there is a triangle of edges between vertices $i$, $a_{i,\alpha}$,$c_{i,\alpha}$. 
\end{enumerate}

\begin{figure}
    \centering
    \includegraphics[scale=0.8]{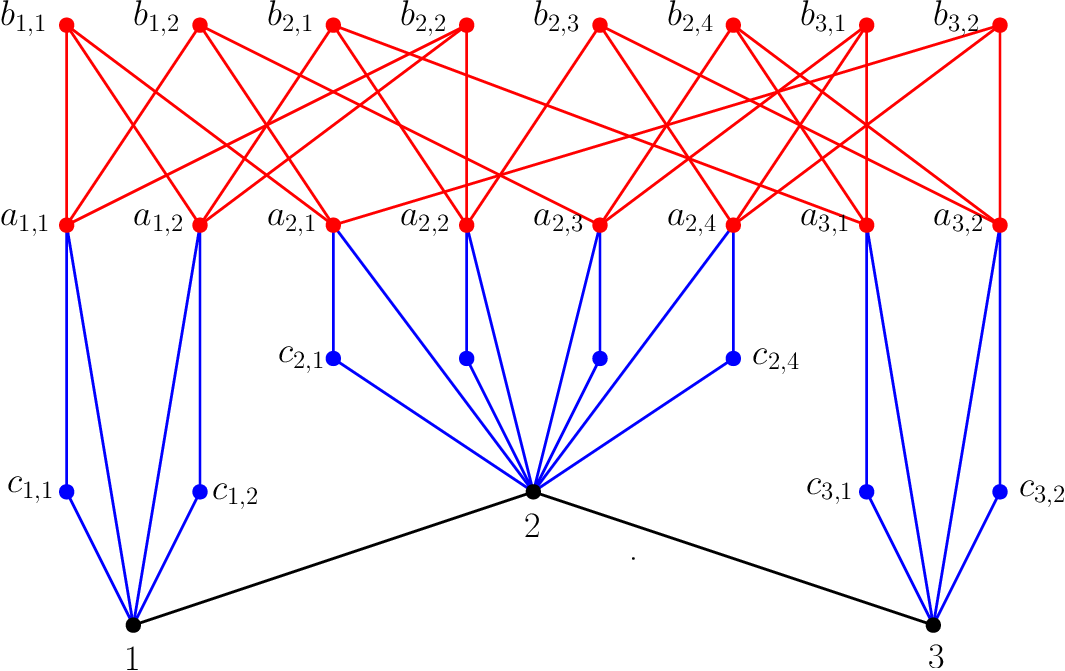}
    \caption{Example of the construction of $G'$ described in Section~\ref{sec:construction}. Here, the original graph $G$ is a path graph on three vertices $\{1,2,3\}$ and the parameter $\eta=2$. The original edges and vertices of $G$ are coloured black; the edges in the triangle gadgets are colooured blue, and the edges in the $d$-regular bipartite graph are coloured red.}
\end{figure}

\subsection{Triangle gadget}
First we consider a triangle gadget, consisting of three vertices $0,1,2$ with edges between all pairs.

\begin{lemma}
    \label{lem:triangle}
    Let $G_{\Delta}=(V,E)$ where $V=\{0,1,2\}$ and $E=\{01,12,02\}$.
    Let $\vvec{x}=(\vec{x_0},\vec{x_1},\vec{x_2})$ be an assignment of unit vectors to the vertices, and suppose $\vec{x_0} \cdot\vec{x_1}=\cos\theta$ for some angle $\theta\in[0,\pi)$.
     Then 
     \[\max_{\vec{x_2}}F^{G_{\Delta}}_k(\vvec{x})
     %= \frac{1}{2}(3+\sqrt{2+2\cos\theta}-\cos\theta )
     =\frac{9}{4}-\frac{1}{4}(2\cos\tfrac{\theta}{2}-1)^2
     \le  \frac{9}{4}-\frac{9}{16\pi^2}(\theta-\tfrac{2\pi}{3})^2\]
\end{lemma}

\begin{proof}
    \begin{align*}
        F^{G_{\Delta}}_k(\vvec{x})
        &=\frac{1}{2}\sum_{uv \in E} 1-\vec{x_u}\cdot \vec{x_v}\\
        &=\frac{1}{2}(3 -\vec{x_2}\cdot(\vec{x_0}+\vec{x_1})- \vec{x_0}\cdot \vec{x_1})
    \end{align*}
This expression is maximised by choosing $\vec{x_2}$ in the $-(\vec{x_0}+\vec{x_1})$ direction. i.e. by setting $\vec{x_2}=-\frac{\vec{x_0}+\vec{x_1}}{\norm{\vec{x_0}+\vec{x_1}}}$. Substituting this in above gives 
\[\max_{\vec{x_2}}F^{G_{\Delta}}_k(\vvec{x})
=\frac{1}{2}(3 +\norm{\vec{x_0}+\vec{x_1}}- \vec{x_0}\cdot \vec{x_1})\]
Note that the above substitution is only possible if $\vec{x_0} +\vec{x_1} \neq 0$. If $\vec{x_0} +\vec{x_1} =0$, then $F^{G_{\Delta}}_k(\vvec{x})$ is independent of $\vec{x_2}$, and the above expression for $\max_{\vec{x_2}}F^{G_{\Delta}}_k(\vvec{x})$ still holds.

Now use the fact that $\norm{\vec{x_0}+\vec{x_1}}^2=(\vec{x_0}+\vec{x_1})\cdot (\vec{x_0}+\vec{x_1}) = 2+2\vec{x_0}\cdot \vec{x_1}$ to get 
\[\max_{\vec{x_2}}F^{G_{\Delta}}_k(\vvec{x})
=\frac{1}{2}(3 +\sqrt{2+2\vec{x_0} \cdot\vec{x_1}}- \vec{x_0}\cdot \vec{x_1}).\]
Substitute in $\vec{x_0} \cdot\vec{x_1}=\cos\theta$ and rearrange to show 
\[\max_{\vec{x_2}}F^{G_{\Delta}}_k(\vvec{x})=\frac{1}{2}(3+\sqrt{2+2\cos\theta}-\cos\theta)=\frac{9}{4} -\frac{1}{4}(\sqrt{2+2\cos\theta}-1)^2\]
Finally, substitute in the half-angle formula $\sqrt{2+2\cos\theta} =2\cos\frac{\theta}{2}$ to obtain the equality stated in the lemma.
%\[\max_{\vec{x_2}}F^G_k(\vvec{x})=\frac{9}{4}-\frac{1}{4}(2\cos\tfrac{\theta}{2}-1)^2.\]

To prove the inequality, it remains to show that $(2\cos\frac{\theta}{2}-1)^2 \ge (1-\frac{3\theta}{2\pi})^2$ for all $\theta \in [0,\pi)$.

Note that $2\cos\frac{\theta}{2}-1$ and  $1-\frac{3\theta}{2\pi}$ are equal at $\theta=0$ and $\theta=2\pi/3$. Since $\frac{d^2}{d\theta^2} (2\cos\frac{\theta}{2}-1)=-\frac{1}{2}\cos\frac{\theta}{2}\le 0$ for $\theta \in [0,\pi)$ and $1-\frac{3\theta}{2\pi}$ is linear, it must be the case that 
\[2\cos\tfrac{\theta}{2}-1 \ge 1-\tfrac{3\theta}{2\pi} \ge 0 \quad \text{for } \theta \in [0,2\pi/3]\]
\[2\cos\tfrac{\theta}{2}-1 \le 1-\tfrac{3\theta}{2\pi} \le 0 \quad \text{for } \theta \in [2\pi/3,\pi]\]
and therefore $(2\cos\frac{\theta}{2}-1)^2 \ge (1-\frac{3\theta}{2\pi})^2$ for all $\theta \in [0,\pi)$ as required.
\end{proof}

\subsection{Bipartite graphs}

\begin{lemma}
    
    \label{lem:bipartitegraph}
    Let $G$ be a $d$-regular connected, bipartite graph with edges across the bipartition $A\cup B=V$.  Let $\lambda_{\max}$ be the largest eigenvalue of the Laplacian $L$ of $G$, and suppose the gap between the largest and second largest eigenvalue is at least $\Delta$.

    Let $\vvec{x}=(\vec{x_1},\vec{x_2},\dots ,\vec{x_n})$ be an assignment of $k$-dimensional vectors to the $n$ vertices of $G$.
    %Then there exists an efficiently computable vector $\vec{z}$ such that $\vec{x_i}=\vec{z}+\vec{e_i}$ and
    Let $\vec{z}=(\sum_{i\in A} \vec{x_i}-\sum_{i \in B} \vec{x_i})/n$, and for each $i\in A$ let $\epsilon_i$ be the angle between $\vec{z}$ and $\vec{x_i}$.
    Then
    %Let $\vec{z}=\sum_{i=1}^n e^*_i\vec{x_i}/n$
    %\[F^G_k(\vvec{x})\le \lambda_{\max} (n-\gamma \sum_{i=1}^n\norm{\vec{e_i}}^2 )\]
    \[F^G_k(\vvec{x})\le \frac{\lambda_{\max} n}{4}-\frac{\Delta}{4\pi^2} \sum_{i \in A}\epsilon_i^2 
    \]
\end{lemma}

In the construction for Lemma~\ref{lem:main}, Lemma~\ref{lem:bipartitegraph} will be applied to a $d$-regular bipartite expander graph, for $d$ constant and with constant spectral gap. Since the spectrum of the adjacency matrix of a bipartite graph is symmetric about 0, this implies that the largest eigenvalue $\lambda_{\max}$ of the Laplacian $L$ is $2d=O(1)$ and the gap $\Delta$ to the second largest eigenvalue is also constant.

\begin{proof}
The total energy is 
\[F^G_k(\vvec{x})=\sum_{ij \in E}\frac{1}{2}(1-\vec{x_i} \cdot \vec{x_j})=\sum_{i,j \in V}\frac{1}{4}L_{ij}\vec{x_i}\cdot \vec{x_j}.\]
Let the coordinates of $\vec{x_i}$ be $x_{il}$ for $l \in [k]$, so that $\vec{x_i}$ can be identified with the $i$th row of the matrix with entries $\{x_{il}\}_{i\in[n], l \in [k]}$. Define the vector $\vec{\hat{x}_l}$ to be the $l$th column of this matrix, i.e. the $n$-dimensional vector with coordinates $x_{il}$ for $i\in [n]$
\[\sum_{i,j \in V}\frac{1}{4}L_{ij}\vec{x_i}\cdot \vec{x_j}
=\sum_{l=1}^k\sum_{i,j \in V}\frac{1}{4}x_{il}L_{ij}x_{jl}
=\sum_{l=1}^k \frac{1}{4}\vec{\hat{x}_l}^{\T}L\vec{\hat{x}_l}.\]

If $G$ is connected then the all ones vector is the unique smallest eigenvector of $L$, with eigenvalue $0$. Since $G$  bipartite across the split $V=A\cup B$, the spectrum of the adjacency matrix $Adj(G)$ is symmetric about 0. To see this note that the diagonal matrix with entry at $i,i$ on the diagonal equal to $+1$ for $i \in A$, and $-1$ for $i \in B$, maps eigenvectors of $Adj(G)$ of eigenvalue $\lambda$ to eigenvectors of eigenvalue $-\lambda$. 
Since $G$ is $d$-regular, the spectrum of $L=dI-Adj(G)$ is therefore symmetric about $d$ and $L$ has a unique largest eigenvalue $\lambda_{\max}=2d$ with corresponding eigenvector $\vec{e^*}$:
\[e^*_i= \begin{cases}
    1 & \text{ for } i \in A \\
    -1 & \text{ for } i \in B.
\end{cases}
\]
Decompose $\vec{\hat{x}_l}$ in terms of $\vec{e^*}$ and some $\vec{\hat{e}_l}$ orthogonal to $\vec{e^*}$
\begin{equation}
    \label{eq:xhatl}
    \vec{\hat{x}_l}=z_l\vec{e^*}+\vec{\hat{e}_l}.
\end{equation}
The vector $\vec{e^*}$ is the unique eigenvector of $L$ corresponding to the largest eigenvalue $\lambda_{\max}$ and so  $\vec{e^{*}}^\T L \vec{e^*}=\lambda_{\max}\norm{\vec{e^*}}^2$. Since $\vec{\hat{e}_l}$ is orthogonal to $\vec{e^*}$, it has overlap with no eigenvector of eigenvalue larger than the second largest eigenvalue $\lambda_{\max}-\Delta$, so $\vec{\hat{e}_l}^\T L \vec{\hat{e}_l}\le (\lambda_{\max}-\Delta)\norm{\vec{\hat{e}_l}}^2$. 
Therefore
\begin{align*}
    F^G_k(\vvec{x})=\sum_{l=1}^k \frac{1}{4}\vec{\hat{x}_l}^{\T}L\vec{\hat{x}_l}
&=\frac{1}{4}\sum_{l=1}^k \left(z_l^2 \vec{e^{*}}^\T L \vec{e^*} +\vec{\hat{e}^\T_l} L \vec{\hat{e}_l}\right)\\
&\le\frac{1}{4}\sum_{l=1}^k \left(z_l^2 \lambda_{\max}\norm{\vec{e^*}}^2 +(\lambda_{\max}-\Delta)\norm{\vec{\hat{e}_l}}^2\right)\\
&=\frac{1}{4}\left(\lambda_{\max}\Big(\sum_{l=1}^k \norm{\vec{\hat{x}_l}}^2\Big)-\Delta\Big(\sum_{l=1}^k\norm{\vec{\hat{e}_l}}^2\Big) \right)
%&=\frac{1}{4}\left(\lambda_{\max}(\sum_{l=1}^k z_l^2+\norm{\vec{\hat{e}_l}}^2)-\Delta\sum_{l=1}^k\norm{\vec{\hat{e}_l}}^2 \right)
%&=\frac{1}{4}\left(\lambda_{\max}n-\Delta\sum_{l=1}^k\norm{\vec{\hat{e}_l}}^2 \right)
\end{align*}
where in the final equality we have used the fact that $\vec{e^*}$ and $\vec{\hat{e}_l}$ are orthogonal to equate $\norm{\vec{\hat{x}_l}}^2=z_l^2 \norm{\vec{e^*}}^2 +\norm{\vec{\hat{e}_l}}^2$.
Since each $\vec{x_i}$ has unit size, we have:

\[ \sum_{l=1}^k\norm{\vec{\hat{x}_l}}^2=\sum_{i=1}^n\sum_{l=1}^k x_{il}^2
=\sum_{i=1}^n \norm{\vec{x_i}}^2=n .\]

Similarly let the coordinates of $\vec{\hat{e}_l}$ be $e_{il}$ for for $i \in [n]$, so that $\vec{\hat{e}_l}$ can be identified with the $l$th column of the matrix with entries $\{e_{il}\}_{i\in[n], l \in [k]}$. Define the vector $\vec{e_i}$ to be the $i$th row of this matrix, i.e. the $k$-dimensional vector with coordinates $x_{il}$ for $l\in [k]$.

Observe that $\sum_{l=1}^k\norm{\vec{\hat{e}_l}}^2 =\sum_{l=1}^k\sum_{i=1}^n e_{il}^2=\sum_{i=1}^n\norm{\vec{e_i}}^2$, so we have 
\[F^G_k(\vvec{x})\le\frac{1}{4}\left(\lambda_{\max}n-\Delta\sum_{i=1}^n\norm{\vec{e_i}}^2 \right).\]
Noting that $\norm{\vec{e_i}}^2\ge 0$ for all $i \in B$, all that remains is to show that $\norm{\vec{e_i}}^2 \ge \epsilon_i^2/\pi^2$ for all $i \in A$. 
We actually prove the slightly stronger bound of
\[\norm{\vec{e_i}}\ge \begin{cases}
    \sin(\epsilon_i) & \text{ if } \epsilon_i \in [0,\pi/2] \\
    1 & \text{ if } \epsilon_i \in [\pi/2,\pi] .
\end{cases} \]

Since $\vec{e^*}$ and $\vec{\hat{e}_l}$ are orthogonal, we can take the inner product of equation \eqref{eq:xhatl} with $\vec{e^*}$ to get:
\[z_l\norm{\vec{e^*}}^2=\vec{\hat{x}_l}\cdot \vec{e^*}=\sum_{i \in A}x_{il} - \sum_{i \in B}x_{il} .\]
So $\vec{z}=(\sum_{i\in A} \vec{x_i}-\sum_{i \in B} \vec{x_i})/n$ is the $k$-dimensional vector with coordinates $z_l$.
Looking at the $i$th coordinate of $\vec{\hat{x}_l}$ in equation \eqref{eq:xhatl} for $i \in A$, we have:
\[x_{il}=z_{l}+e_{il} \quad \Rightarrow \quad \vec{x_i}=\vec{z}+\vec{e_i} .\]
The angle $\epsilon_i$ between $\vec{z}$ and $\vec{x_i}$ satisfies $\vec{z} \cdot \vec{x_i}=\norm{\vec{z}}\norm{\vec{x_i}}\cos(\epsilon_i)=\norm{\vec{z}}\cos(\epsilon_i)$, and so 
\begin{align}
    \norm{\vec{e_i}}^2&=(\vec{x_i}-\vec{z})\cdot(\vec{x_i}-\vec{z})\\
    &=\vec{x_i}\cdot \vec{x_i} -2 \vec{x_i}\cdot \vec{z}+ \vec{z}\cdot \vec{z} \\
    &=1-2\norm{\vec{z}}\cos(\epsilon_i) +\norm{\vec{z}}^2.
    \label{eq:eisquared}
\end{align}
If $\epsilon_i \in [\pi/2,\pi]$, then $\cos(\epsilon_i)\ge 0$, and so $\norm{\vec{e_i}}^2\ge 1$.
Otherwise, if $\epsilon_i \in [0,\pi/2]$, we can complete the square in equation \eqref{eq:eisquared} to get 
\[\norm{\vec{e_i}}^2=1-\cos^2(\epsilon_i) +\left(\norm{\vec{z}}-\cos(\epsilon_i)\right)^2
\ge \sin^2(\epsilon_i).\]
\end{proof}

\subsection{Proof of Lemma~\ref{lem:main}}
\label{sec:lemproof}
We are now ready to prove Lemma~\ref{lem:main}.
\begin{proof}
    For part 1. of Lemma~\ref{lem:main}, let $\vvec{x}$ be the assignment that achieves $F_k^G(\vvec{x})=\MAXCUT_k(G)$.
Embed each vector $\vec{x_i}$ into $\R^{k+1}$ by padding with an extra $0$ and call the resulting vector $\vec{x'_i}\in \R^{k+1}$. Let $\vec{z}=(0,0,\dots,0,1)\in \R^{k+1}$  be the new direction.

Define $\vvec{y}\in (S_{k+1})^{V'}$ by:
\[\vec{y_i}=\cos(2\pi/3)\vec{z}+\sin(2\pi/3)\vec{x'_i} \]
\[\vec{y_{c_{i,\alpha}}}=\cos(4\pi/3)\vec{z}+\sin(4\pi/3)\vec{x'_i} \]
\[\vec{y_{a_{i,\alpha}}}=\vec{z} \quad \text{ and } \vec{y_{b_{i,\alpha}}}=-\vec{z}.\]
This assignment achieves the maximum possible $9/4$ for each triangle $i,a_{i,\alpha},c_{i,\alpha}$. The expander graph includes $4\eta m$ vertices in total, and so $\vvec{y}$ achieves the maximum of $\lambda_{\max}\eta m$ for the expander graph.
For the edges in $G'$ corresponding to the original edges of $G$:
\[\frac{1}{2}[1-\vec{y_i}\cdot\vec{y_j}]
=\frac{1}{2}[1-\cos^2(2\pi/3)-\sin^2(2\pi/3)\vec{x_i}\cdot \vec{x_j}]=\frac{3}{4}\times \frac{1}{2}[1-\vec{x_i}\cdot \vec{x_j}] .\]
So in total:
\[F_{k+1}^{G'}(\vvec{y})=\frac{9}{4}2\eta m +\lambda_{\max}\eta m+\frac{3}{4}F_{k}^G(\vvec{x}) .\]
Noting that $\MAXCUT_{k+1}(G')\ge F_{k+1}^{G'}(\vvec{y})$ and $F_{k}^G(\vvec{x})=\MAXCUT_k(G)$, we have 
\[\MAXCUT_{k+1}(G') \ge C\eta m +\frac{3}{4} \MAXCUT_k(G)\]
where $C=9/2+\lambda_{\max}$.

\vspace{1cm}
The remainder of the proof focuses on part 2. of Lemma~\ref{lem:main}. Note that we reuse the variable names $\vec{x}$, $\vec{y}$, $\vec{z}$, although they are not the same as in part 1.  
We start by describing the map $g:(S_{k})^{n+6\eta m}\rightarrow (S_{k-1})^n$ which maps inputs for $F_{k+1}^{G'}$ to inputs for $F_k^G$. 

Let $\vvec{y} \in (S_{k})^{n+6\eta m}$ and set  $\vec{z}=\sum_{a \in A} \vec{y_{a}}-\sum_{b\in B} \vec{y_{b}}/\norm{\sum_{a \in A} \vec{y_{a}}-\sum_{b\in B} \vec{y_{b}}}$.
Now change basis so that $\vec{z}=(0,0,\dots,0,1)\in \R^{k+1}$. For each $i\in V$, let $\vec{x_i}$ be the unit vector you get by taking the first $k$ entries of $\vec{y_i}$ in this basis and renormalising.
Let $g(\vvec{y})=\vvec{x}=(\vec{x_1},\vec{x_2},\dots \vec{x_n})$.

We write $\vec{x'}$ for the embedding of $\vec{x}$ into $\R^{k+1}$ by padding with an extra zero; or equivalently $\vec{x_i'}$ is the (renormalised) projection of $\vec{y_i}$ onto the space orthogonal to $\vec{z}$. 
Since $\vec{y_i}$, $\vec{z}$ and $\vec{x_i'}$ are all unit vectors for all $i\in V$, we have:
\[\vec{y_i} =\cos(\theta_i) \vec{z} +\sin(\theta_i)\vec{x'_i}\]
for some angle $\theta_i$ between $\vec{z}$ and $\vec{y_i}$.

We can now bound the contribution from different parts of $G'$. First consider the expander graph across $A$ and $B$.
Let $\epsilon_{i,\alpha}$ be the angle between $\vec{y_{a_{i,\alpha}}}$ and $\vec{z}$. 
By Lemma~\ref{lem:bipartitegraph}, the contribution to $F_{k+1}^{G'}(\vvec{y})$ from the bipartite graph is at most 
\[\lambda_{\max}\eta m-\frac{\Delta}{4\pi^2}\sum_{i,\alpha}\epsilon_{i,\alpha}^2.\]

Second, consider the triangle gadgets. The angle between $\vec{y_i}$ and $\vec{y_{a_{i,\alpha}}}$ is in $[\theta_i-\epsilon_{i,\alpha}, \theta_i+\epsilon_{i,\alpha}]$. Let $\phi_i=\frac{2\pi}{3}-\theta_i$ so that by Lemma~\ref{lem:triangle}, the contribution from the triangle $i, a_{i,\alpha}, c_{i,\alpha}$ is at most
\[\frac{9}{4}-\frac{9}{16\pi^2}(\phi_{i}+\epsilon_{i,\alpha})^2.\]

Third, consider the original edges from $G$ between the vertices in $V$. For an edge $ij \in E$, if $\theta_i=\theta_j=2\pi/3$ (i.e. $\phi_i=0=\phi_j$), then the contribution of that edge to $F_{k+1}^{G'}(\vvec{y})$ is 
\[\frac{3}{8}(1-\vec{x_i}\cdot \vec{x_i}).\]
If $\phi_i$ and $\phi_j$ are non-zero, this can change the angle $\xi_i$ between $\vec{y_i}$ and $\vec{y_j}$ by at most $|\phi_i|+|\phi_j|$.
Since $\vec{y_i}\cdot\vec{y_j}=\cos(\xi_i)$, the rate of change with respect to $\xi_i$ is at most $|\frac{d}{d{\xi_i}} \cos(\xi_i) |=|\sin(\xi_i)|\le 1$.
Therefore the edge $ij$ can contribute at most 
\[\frac{3}{8}(1-\vec{x_i}\cdot \vec{x_i})+|\phi_i|+|\phi_j|.\]
Adding all these contributions together gives:

\[F_{k+1}^{G'}(\vvec{y}) \le \lambda_{\max} \eta m -\frac{\Delta}{4\pi^2}\sum_{i,\alpha}\epsilon_{i,\alpha}^2
+\sum_{i,\alpha}\left(\frac{9}{4}-\frac{9}{16\pi^2}(\phi_{i}+\epsilon_{i,\alpha})^2\right)\]
\[\qquad \qquad \qquad +\sum_{ij \in E}\left(\frac{3}{8}(1-\vec{x_i}\cdot \vec{x_i})+|\phi_i|+|\phi_j|\right)\]

\[\qquad \qquad =C\eta m +\frac{3}{4}F_k^G(\vvec{x})-\sum_{i,\alpha}\left(c_1\epsilon_{i,\alpha}^2+c_2(\phi_{i}+\epsilon_{i,\alpha})^2\right) +\sum_{i\in V} d_i |\phi_i|\]
where $C=9/2+\lambda_{\max}$, $c_1=\frac{\Delta}{4\pi^2}$ and $c_2=\frac{9}{16\pi^2}$.
If the bipartite expander graph is chosen to be low degree and with good spectral expansion such that $\lambda_{\max}$ and $\Delta$ are constant, then $C,c_1,c_2$ are all constant.

Looking at the terms with $\epsilon_{i,\alpha_i}$ and rearranging we have
\[c_1\epsilon_{i,\alpha}^2+c_2(\phi_{i}+\epsilon_{i,\alpha})^2=(c_1+c_2)(\epsilon_{i,\alpha}+\frac{c_2}{c_1+c_2}\phi_i)^2 +\frac{c_1c_2}{c_1+c_2}\phi_i^2 \ge c'\phi_i^2\]
where $c'=\frac{c_1c_2}{c_1+c_2}$. Substituting this back in and summing over $\alpha$, we get:
\[F_{k+1}^{G'}(\vvec{y}) \le C\eta m +\frac{3}{4}F_k^G(\vvec{x})
+\sum_{i\in V} d_i(-\eta c'\phi_i^2 +|\phi_i|).\]
To upper bound $-\eta c'\phi_i^2 +|\phi_i|$, note that it is symmetric about zero, and for $\phi_i\ge 0$, one can complete the square $-\eta c'\phi_i^2 +\phi_i=-\eta c'(\phi_i-\frac{1}{2\eta c'})^2+\frac{1}{4\eta c'} \le \frac{1}{4\eta c'}$.
Putting this back in leads to 
\[F_{k+1}^{G'}(\vvec{y}) \le C\eta m +\frac{3}{4}F_k^G(\vvec{x})+\frac{1}{2c'\eta}m.\]
\end{proof}

\bibliography{maxbib}

\end{document}